\documentclass[a4paper,11pt]{amsart}
\usepackage[utf8]{inputenc}
\usepackage[english]{babel}
\usepackage{amsaddr}
\usepackage{mathrsfs}
\usepackage{amssymb}
\usepackage{hyperref}
\usepackage{graphicx,color}
\usepackage{amsthm}
\usepackage[dvipsnames]{xcolor}
\usepackage[normalem]{ulem}
\usepackage{cancel}
\usepackage{float}
\usepackage{enumitem}
\usepackage{amsmath}
\usepackage{enumitem}
\usepackage{cite}
\newtheorem{theorem}{Theorem}

\newtheorem{lemma}{Lemma}

\theoremstyle{definition}
\theoremstyle{plain}
\theoremstyle{remark}

\newcommand{\R}{\mathrm{R}}

\renewcommand{\H}{\mathcal{H}}

\newcommand{\p}{\partial}
\newcommand{\Keywords}[1]{%
  \par\noindent\textbf{Keywords:} #1
}

\begin{document}
\vspace*{.5cm}

\title{Parameter estimation for the mixed fractional Merton jump diffusion model with EM algorithm}
\date{\today}
\author[
Agboeke, Almani, Gasbarra, Shokrollahi, and Sottinen
]{
Joy Agboeke, \\
Hamidreza Maleki Almani, \\
Dario Gasbarra, \\
Foad Shokrollahi, \\
Tommi Sottinen
}
\address{Department of Mathematics and Statistics,\\ School of Technology and Innovations, University of Vaasa,\\ P.O. Box 700, FIN-65101 Vaasa, FINLAND}

\begin{abstract}
This paper proposes an Expectation--Maximization algorithm with Metropolis--Hastings sampling for parameter estimation in a Mixed Fractional Merton Jump Diffusion model. The model combines fractional Brownian motion to capture long-range dependence with a compound Poisson jump process to describe abrupt movements in financial returns. The latent jump process is inferred during the E-step using Markov Chain Monte Carlo sampling, while the M-step updates the model parameters by maximizing the expected complete-data likelihood. The consistency and asymptotic normality of the proposed estimator are established under suitable regularity conditions. The methodology is applied to the Helsinki Stock Index (OMXH25). The proposed estimation framework provides a reliable and computationally efficient approach for modeling financial time series exhibiting both long-memory dependence and jump behavior.
\end{abstract}
\Keywords: {Fractional Brownian motion, Jump diffusion, EM algorithm, MCMC, Metropolis--Hastings, Parameter estimation, Financial time series, Long memory,Merton model, Monte Carlo simulation}


\maketitle

\vspace*{1cm}
\section{Introduction}

Merton \cite{merton1976option} extended the classical geometric Brownian motion model by incorporating jump components to account for sudden discontinuities in asset prices. The resulting jump-diffusion model is given by
\begin{equation}
\label{eq:1}
S_t = S_0 \exp\left(\mu t +\sigma W + J_t\right).
\end{equation}
where $W$ is a standard Brownian motion and $J_t = \sum_{i=1}^{N_t} Z_i$ is a compound Poisson process with intensity $\lambda > 0$. This model captures excess kurtosis and heavy tails observed in financial returns, which are not explained by the Black--Scholes framework \cite{Merton1976}.

Empirical evidence suggests that financial time series exhibit long-range dependence and persistence Volatility see\cite{longmemory2002,sadique2001}. To model such behavior, Mandelbrot and Van Ness \cite{Mandelbrot1968} introduced fractional Brownian motion (fBm) $B_t^H$, a centered Gaussian process with covariance function
\begin{equation}
\label{eq:2}
\mathbb{E}[B_t^H B_s^H] = \frac{1}{2}\left(t^{2H} + s^{2H} - |t-s|^{2H}\right),
\end{equation}
where $H \in (0,1)$ is called the Hurst parameter. fBm exhibits short memory when $H < 1/2$, long memory when $H >1/2$ and reduces to the standard Brownian motion when $H = 1/2$ \cite{Mandelbrot1968}. fBm has attracted considerable attentions in the literature due to its ability to capture long-range dependence and self-similar structures \cite{Franzke2020, Beran1994,Samorodnitsky2006}. Several studies have also investigated its theoretical properties and applications in stochastic modeling and finance \cite{hu2003,duncan2000,mishura2008}. However, for $H \neq 1/2$, fBm is not a semi-martingale \cite{sottinen2001}, limiting its direct application in financial modeling.

To overcome this limitation, Cheridito \cite{Cheridito2001} introduced the mixed fractional Brownian motion (mfBm) which is defined as a linear combination of a standard Brownian motion and an independent fBm, which preserves the semi-martingale property under suitable conditions. mfBm, its structural properties, together with parameter estimation methodologies, have been extensively investigated in the literature \cite{hu2011,xiao2015,masnita2010,xiao2011,sun2018,Beran1994,fay2009,palma2007}. These studies demonstrate that mixed fractional processes are capable of capturing both short-term fluctuations and long-range dependence commonly observed in financial time series.

Despite these developments, empirical applications have highlighted certain limitations of pure mfBm when modeling abrupt market movements. In particular \cite{dufitinema2020} applied mixed fractional Brownian motion to Finnish stock market data and observed that the discontinuities and extreme price movements remain insufficiently captured without incorporating jump dynamics.

Motivated by these findings, the present study extends the mixed fractional framework by incorporating jump components into the price dynamics. This approach inputs the price memory into the Merton model~\ref{eq:1} by adding an fBm component to its exponential form. Combining diffusion, long memory, and jump components which leads to the mixed fractional Merton jump diffusion (MFMJD) model:
\begin{equation}
\label{eq:3}
S_t = S_0 \exp\left(\mu t +\sigma W + \nu B_t^{H} + J_t\right).
\end{equation}
where $\mu \in \mathbb{R}$ is the drift parameter, $\sigma^2 $ is the diffusion volatility, $\nu^2 $ is the fractional volatility coefficient, $W$ is a standard Brownian motion, $B_t^H$ is an fBm with Hurst parameter $H \in (1/2,1)$, and $J_t = \sum_{i=1}^{N_t} Z_i$ is a compound Poisson process with intensity $\lambda > 0$, where $N_t \sim \text{Poisson}(\lambda \Delta t)$ and the jumps $\{Z_i\}$ are i.i.d. 
${\mathcal N}( \mu_J,\sigma_J^2)$. For analytical construction and application of ~\ref{eq:1}, see \cite{merton1976option}

In discrete time, the log-return process $Y_t = \Delta\log S_t $ can be expressed as
\begin{equation*}
Y_t = \mu \Delta t + J_t + \varepsilon_t,
\end{equation*}
where $N_t \sim \text{Poisson}(\lambda \Delta t)$ is the latent jump count process, and $\varepsilon_t$ is a Gaussian process with covariance structure induced by the mixed fractional component.

The presence of latent jump variables renders the likelihood function analytically intractable. To address this, we use the the Expectation-Maximization (EM) algorithm for parameter estimation in the presence of missing data introduced by Dempster, Laird, and Rubin \cite{Dempster1977} . In this framework, the jump process $\{N_t\}$ is treated as unobserved data. Also, the conditional distribution of $N_t$ given observations is not available in closed form. To overcome this difficulty, simulation-based methods are employed. Recall that, the Metropolis algorithm \cite{Metropolis1953}, later generalized by Hastings \cite{Hastings1970}, provides a mechanism for sampling from complex posterior distributions. The combination of EM with Metropolis-Hastings(MH) sampling, commonly referred as Monte Carlo EM, has been successfully applied in jump-diffusion and latent variable models \cite{Ball1985, Johannes2009}.

In this paper, we apply this framework to the mixed fractional Merton jump diffusion model. The E-step is approximated using Markov chain Monte Carlo (MCMC) methods, specifically the Metropolis--Hastings algorithm, to generate samples from the conditional distribution of the latent jump process. MCMC methods construct a Markov chain whose stationary distribution is the target distribution, allowing expectations with respect to an analytically intractable distribution to be approximated by empirical averages of simulated draws \cite{RobertCasella2004,BrooksGelmanJonesMeng2011}. In the present setting, this enables the unobserved jump process to be integrated into the E-step through Monte Carlo approximation, while the M-step updates the model parameters by maximizing the resulting expected complete-data log-likelihood. This approach enables joint estimation of the model parameters and latent jumps in the presence of both long memory and discontinuities, and motivates the parameter identification procedure presented in the next section.

The remainder of this paper is organized as follows: Section 2 presents the MFMJD together with the corresponding parameter identification framework. Developing the proposed EM-MH estimation methodology by deriving the complete-data likelihood, the score equations, and the EM updating scheme. In section 3,  the theoretical properties of the proposed estimator are then established through consistency and asymptotic normality results under appropriate regularity conditions. Section 4 describes the data and it's properties while Section 5 illustrates the application of the proposed methodology using real financial data. The paper concludes with a summary of the principal findings and suggestions for future research.

\section{Preliminaries and Methodology}
\subsection{Model Specification and Parameter Identification}\mbox{}\\

The following parameters are identified from equation ~\ref{eq:3} above: 
\begin{enumerate}[label=(\roman*)]
\item $\mu \in \mathbb{R}$ is the drift parameter.
\item $\sigma^2 > 0$ is the Brownian volatility,
\item $\nu > 0$ is the fractional volatility coefficient,
\item $H \in (1/2,1)$ controls long-range dependence,
\item $\lambda > 0$ is the jump intensity,
\item $\mu_J > 0$ is the mean of the jumps,
\item $\sigma^2_J > 0$ is the jump volatility coefficient,
\end{enumerate}
We model the return process as:
\begin{equation*}
	\mathbf{Y} = \mu \Delta t \mathbf{1} + \varepsilon,
\end{equation*}
where {$\mathbf{1}=(1,\ldots,1)^\top\in\R^n$, and} $\varepsilon$ captures both continuous and jump-induced variability.

The covariance of the Gaussian component is given by:
\begin{equation*}
	\Sigma = \sigma^2 \Delta t\, I + \nu^2 \Gamma(H),
\end{equation*}
where $\Gamma_H$ is the Toeplitz matrix induced covariance structure of fBm given by ~\ref{eq:2}:
\begin{equation*}
(\Gamma_H)_{ij} = \frac{1}{2}\left(|i-j+1|^{2H} + |i-j-1|^{2H} - 2|i-j|^{2H}\right)\Delta t^{2H}.
\end{equation*}
We introduce a state-dependent conditional covariance structure driven by latent jumps:
\begin{equation*}
	\Sigma = \Sigma + \sigma_J^2 \,\mathrm{diag}(\mathbf{N}),
\end{equation*}
where $\mathbf{N} = (N_1,\dots,N_n)^\top$ represents latent jump counts,{and $N_i$ is the number of jumps on $(t_{i-1}, t_i]$.}
Thus, conditional on $\mathbf{N}$, the return vector follows:
\begin{equation*}
	\mathbf{Y} \mid \mathbf{N} \sim \mathcal{N}(\mu \Delta t \mathbf{1}+\mu_J N, \Sigma).
\end{equation*}
This specification implies that jump activity induces local increases in volatility. Under this formulation, the parameter vector
\begin{equation}
 \label{eq:4}   
\theta = (\mu,\lambda, \sigma^2, \nu^2, H,\sigma^2_J, \mu_J))^\top
\end{equation}
is identifiable provided that $\Sigma$ is positive definite.

The complete-data likelihood is:
\begin{equation*}
	L(\theta; \mathbf{Y}, \mathbf{N}) = p(\mathbf{Y} \mid \mathbf{N}; \theta)\, p(\mathbf{N}; \theta),
\end{equation*}
where for the gaussian component 
\begin{equation*}
	p(\mathbf{Y} \mid \mathbf{N}) =
	\frac{1}{(2\pi)^{n/2} |\Sigma|^{1/2}}
	\exp\left(
	-\frac{1}{2}
	(\mathbf{Y} - \mu \Delta t \mathbf{1} -\mu_J N)^\top
	\Sigma^{-1}
	(\mathbf{Y} - \mu \Delta t \mathbf{1} -\mu_J N)
	\right).
\end{equation*}
and for the jump part
\begin{equation*}
	p(\mathbf{N}) =
	\prod_{i=1}^n
	\frac{(\lambda \Delta t)^{N_i}}{N_i!}
	e^{-\lambda \Delta t}.
\end{equation*}
The complete-data log-likelihood is:
\begin{equation*}
	\begin{aligned}
		\ell_n(\theta; Y,N)
		&=
		-\frac{1}{2}
		(\mathbf{Y} - \mu \Delta t \mathbf{1} - \mu_J \mathbf{N} )^\top
		\Sigma^{-1}
		(\mathbf{Y} - \mu \Delta t \mathbf{1} - \mu_J \mathbf{N}) \\
		&\quad
		-\frac{1}{2} \log \det(\Sigma) \\
		&\quad
		+ \sum_{i=1}^n
		\left[
		N_i \log(\lambda \Delta t)
		- \lambda \Delta t
		- \log(N_i!)
		\right].
	\end{aligned}
\end{equation*}
We need to maximize w.r.t.
$\theta$
\begin{align*}
E_{\theta_0}(\ell_n(\theta;Y,N)| Y)
\simeq \frac{1}{M}
\sum_{k=1}^M \ell_n(\theta;Y,N^{(k)})
\end{align*}
where $(N^{(k)}: k=1,\dots,M)$
is the MCMC conditioned on $Y$.

The score vector is 
$\mathbf{S}(\theta)=\nabla_\theta E(\ell_n|Y)$, i.e.,
	\begin{gather*}
		\mathbf{S}=(S_1,\ldots,S_6)^\top\\ 
		= (\p_\lambda,\p_\mu,\p_{\sigma^2},\p_{\nu^2},\p_H,\p_{\sigma_J^2},\p_{\mu_J})\ell_n.\notag
    \end{gather*} 
Let $\mathbf{r(N)} = \mathbf{Y} - \mu \Delta t \mathbf{1} - \mu_J\mathbf{N}$.

For the jump Intensity $\lambda$
\begin{equation*}
	\frac{\partial E(\ell_n|Y)}{\partial \lambda}
	=
	\sum_{i=1}^n \left( \frac{E(N_i|Y)}{\lambda} - \Delta t \right),
\end{equation*}
leading to:
\begin{equation*}
	\hat{\lambda} = \frac{\sum_{i=1}^n E(N_i|Y)}{n \Delta t}.
\end{equation*}

For the drift Parameter $\mu$
\begin{equation*}
	\frac{\partial E(\ell_n|Y)}{\partial \mu}
	=
	(\Delta t \mathbf{1})^\top E\bigl(\Sigma^{-1} \mathbf{r(N)}
    |Y).
\end{equation*}
Setting this to zero yields:
\begin{equation*}
	\hat{\mu} =
	\frac{E\bigl(\mathbf{1}^\top \Sigma^{-1} (\mathbf{Y} -\mu_J\mathbf{N}) \big\vert Y\bigr)}
	{\Delta t\, E(\mathbf{1}^\top \Sigma^{-1} \mathbf{1}| Y)}.
\end{equation*}

For the jump mean $\mu_J$
\begin{equation*}
	\frac{\partial E(\ell_n|Y)}{\partial \mu_J}
	=
	 E(\mathbf{N}^\top \Sigma^{-1} \mathbf{r(N)}|Y).
\end{equation*}
Setting this to zero yields:
\begin{equation*}
	\hat{\mu}_J
	=
	\frac{E\bigl( \mathbf{N}^\top \Sigma^{-1} (\mathbf{Y} -\mu\Delta t\mathbf{1})
    \big\vert Y\bigr)
    }
	{E\bigl(\mathbf{N}^\top \Sigma^{-1} \mathbf{N}
    \vert Y)}.
\end{equation*}
For the jump variance $\sigma_J^2$, 
since it enters the covariance matrix $\Sigma$, the score function is:
\begin{align*}
	\frac{\partial E(\ell_n|Y)}{\partial \sigma_J^2}
	=
	-\frac{1}{2}\operatorname{tr}E\left(\Sigma^{-1} \mathrm{diag}(\mathbf{N}) \vert Y\right)\\
	+
\frac{1}{2} E\bigl(\mathbf{r(N)}^\top\Sigma^{-1}
\mathrm{diag}(\mathbf{N})\Sigma^{-1}\mathbf{r(N)}
\big\vert Y\bigr).
\end{align*}
This equation does not admit a closed-form solution and is handled implicitly via EM updates.

For the parameters of the gaussian part  $q \in \{\sigma^2, \nu^2, H\}$:
\begin{align*}
		\frac{\partial E(\ell_n| Y)}{\partial q}
		=
		-\frac{1}{2} \operatorname{tr}E(\Sigma^{-1}\vert  Y
        ) \partial_q \Sigma
        +
		\frac{1}{2}E \bigl( 
	\mathbf{r(N)}^\top
		\Sigma^{-1}
		(\partial_q \Sigma)
		\Sigma^{-1}
		\mathbf{r(N)}
        \big\vert Y\bigr).
\end{align*}
Due to the nonlinear dependence of $\Gamma(H)$, we estimate these parameters numerically.
Indeed, we have
	\begin{align*}
		\frac{\partial E(\ell_n| Y)}{\partial \sigma^2}
		&=
		\frac{\Delta t}{2}E\left[\mathbf{r(N)}^\top
		\Sigma^{-2}
		\mathbf{r(N)} - \operatorname{tr}(\Sigma^{-1})
        \big\vert Y
        \right],\\
		\frac{\partial E(\ell_n|Y)}{\partial \nu^2}
		&=
		\frac{1}{2}E\left[\mathbf{r(N)}^\top
		\Sigma^{-1}\Gamma(H)\Sigma^{-1}
		\mathbf{r(N)} - \operatorname{tr}\left(\Sigma^{-1}\Gamma(H)\right)
        \big\vert Y
        \right],\\
		\frac{\partial E(\ell_n|Y)}{\partial H}
		&=
		\frac{\nu^2}{2}E\left[\mathbf{r(N)}^\top
		\Sigma^{-1}\partial_H\Gamma(H)\Sigma^{-1}
		\mathbf{r(N)} - \operatorname{tr}\left(\Sigma^{-1}\partial_H\Gamma(H)\right)\big\vert Y\right].
	\end{align*}
\subsection{EM Algorithm with MH Sampling}

\paragraph{\textbf{E-step}}

We compute:
\begin{equation*}
	Q(\theta \mid \theta^{(k)})
	=
	\mathbb{E}_{\mathbf{N} \mid \mathbf{Y}, \theta^{(k)}}
	\left[
	\ell_n(\theta; \mathbf{Y}, \mathbf{N})
	\right].
\end{equation*}
Since the posterior
\begin{equation*}
	p(\mathbf{N} \mid \mathbf{Y})
	\propto
	p(\mathbf{Y} \mid \mathbf{N}) p(\mathbf{N})
\end{equation*}
is analytically intractable due to the covariance term $\Sigma$, we approximate it using a MH sampler.
\paragraph{\textbf{MH Sampling}}
We construct a Markov chain over $\mathbf{N}$ using local and block proposals:
\begin{enumerate}[label=(\roman*)]
	\item Local moves: $N_i \to N_i \pm 1$ with probability 1/2 and a deterministic proposal $0\mapsto 1$ when
$N_i=0$
	\item Block moves: simultaneous updates on randomly selected subsets
\end{enumerate}

The acceptance probability is:
\begin{equation*}
	\alpha = \min\left(1,
	\frac{p(\mathbf{Y} \mid \mathbf{N}^\prime)p(\mathbf{N}^\prime)}
	{p(\mathbf{Y} \mid \mathbf{N})p(\mathbf{N})}
	\right).
\end{equation*}
Samples are collected after burn-in and thinning, and the posterior expectation is approximated by:
\begin{equation*}
	\hat{\mathbf{N}} = \frac{1}{M} \sum_{m=1}^M \mathbf{N}^{(m)}.
\end{equation*}
For the M step: 
Given $\hat{\mathbf{N}}$, parameters are updated by:
\begin{enumerate}[label=(\roman*)]
	\item Closed-form updates  $
	\hat{\mu},\hat{\mu}_J,\hat{\lambda}$
	\item Numerical optimization:  $
	\hat{\sigma}^2, \hat{\nu}^2,\hat{H}, \hat{\sigma}_J^2 $
\end{enumerate}

As covariance criterion, the EM algorithm is iterated until:
\begin{equation*}
	\max_j \frac{|\theta_j^{(k+1)} -\theta_j^{(k)}|} {|\theta_j^{(k)}| }
	< \text{tol}.
\end{equation*}
Where tol is a small number. In this study, $tol = 0.0001$ The algorithm terminates when relative parameter changes fall below a prescribed tolerance $(0.0001)$ .

This formulation introduces a jump-driven heteroskedastic structure through the covariance matrix $\Sigma$. Unlike classical Merton models where jumps affect only the mean, this specification allows jump events to directly influence volatility, thereby capturing volatility clustering and extreme event amplification.

The EM--MH framework enables efficient estimation in the presence of latent jump processes and covariance structures, combining analytical updates with stochastic simulation.
\section{Asymptotic Properties}
For the parameter vector ~\ref{eq:4}, we consider the long-span asymptotic regime
$n\to\infty$, with $\Delta t>0$ fixed

Let
$
Y=(Y_1,\ldots,Y_n)^\top
$
denote the observed return vector and
$
N=(N_1,\ldots,N_n)^\top
$
the latent jump-count vector. The MFMJD model is rewritten as
$$
Y_t = \mu\Delta t+X_t+\mu_JN_t +\sigma_J\sqrt{N_t}Z_t,
\qquad t=1,\ldots,n,
$$
where
$X = \sigma W + \nu B^H$, 
$
N_t\sim\operatorname{Poisson}(\lambda\Delta t), 
 Z_t\stackrel{\mathrm{iid}}{\sim}N(0,1),
$
and the jump process, the Gaussian innovations ${Z_t}$, and the fractional Gaussian noise component are mutually independent. Conditional on $N$, the observations are Gaussian:
$$
Y\mid N \sim N\left(\mu\Delta t\,\mathbf{1}_n+\mu_JN,\Sigma_n(\theta)\right),
$$
where
$$
\Sigma_n(\theta) = \sigma^2\Delta t I_n
+
\nu^2\Gamma_n(H) + \sigma_J^2\operatorname{Diag}(N).
$$
Here $\Gamma_n(H)$ denotes the covariance matrix of the fractional Gaussian noise component.

The asymptotic results below establish consistency and asymptotic normality of the proposed estimator. The long-memory Gaussian component is treated within the framework of strongly dependent stationary Gaussian processes developed by Fox and Taqqu~\cite{FoxTaqqu1986} and Dahlhaus~\cite{Dahlhaus1989}; see also chapter three of Beran~\cite{Beran1994}. The consistency follows from the general theory of
M-estimators; see Theorem 5.7 and the Argmax Theorem
(Theorem 5.9) of \cite{VanDerVaart2000},  while the general asymptotic theory of statistical estimation is given by Ibragimov and Has'minskii Theorems I.5.1–I.5.2; I.10.1–I.10.2 ~\cite{IbragimovHasminskii2013}.

\subsection{Primitive Assumptions}

The following assumptions are imposed directly on the model and parameter space.

\paragraph{(A1) Compact parameter space and interiority.}
The parameter space is

$$
\Theta= \left\{ \theta:
\begin{array}{l}
\mu\in[\underline{\mu},\overline{\mu}],
\quad
\mu_J\in[\underline{\mu}_J,\overline{\mu}_J],
\\[1mm]
\lambda\in[\underline{\lambda},\overline{\lambda}],
\\[1mm]
\sigma^2\in[\underline{\sigma}^2,\overline{\sigma}^2],
\\[1mm]
\nu^2\in[\underline{\nu}^2,\overline{\nu}^2],
\\[1mm]
\sigma_J^2\in
[\underline{\sigma}_J^2,\overline{\sigma}_J^2],
\\[1mm]
H\in
\left[
\dfrac12+\varepsilon_H,
1-\varepsilon_H
\right]
\end{array}
\right\},
$$

where all lower bounds are strictly positive and
$
0<\varepsilon_H<\frac12.
$
The true parameter satisfies
$
\theta_0\in\operatorname{int}(\Theta).
$
Consequently, $\Theta$ is compact.

\paragraph{(A2) Smoothness and integrability.}
For almost every $(Y,N)$, the complete-data log-likelihood
$$
\ell_c(\theta;Y,N)
=
\log p_\theta(Y,N)
$$
is twice continuously differentiable with respect to $\theta$ on $\Theta$. The first- and second-order derivatives are measurable, and there exists an integrable envelope dominating the likelihood and the derivatives required in the arguments below.

\paragraph{(A3) Long-memory spectral regularity.}
The spectral density of the fractional Gaussian noise component satisfies
$$
f_H(\omega)
\sim
C(H)|\omega|^{1-2H},
\qquad
\omega\to0,
$$
where
$
C(H)>0, and  \ H\in
\left[
\frac12+\varepsilon_H,
1-\varepsilon_H
\right].
$

Equivalently,
$$
\lim_{\omega\to0}
\frac{f_H(\omega)}
{C(H)|\omega|^{1-2H}}
=1.
$$
Thus, the Gaussian component has long-range dependence because $H>1/2$. The regularity conditions required for likelihood-based asymptotic analysis of strongly dependent Gaussian processes are assumed to hold; see Theorem 1 in Fox and Taqqu~\cite{FoxTaqqu1986}, see also Dahlhaus~\cite{Dahlhaus1989}, and Beran~\cite{Beran1994}.

\paragraph{(A4) Jump-process regularity.}
The jump counts are iid and satisfy
$
N_t\sim\operatorname{Poisson}(\lambda\Delta t),
$
and
$
Z_t\stackrel{\mathrm{iid}}{\sim}N(0,1).
$
The sequences ${N_t}$ and ${Z_t}$ are independent of each other and of the fGN process. Consequently, the jump component possesses finite moments of every order.

\paragraph{(A5) Monte Carlo approximation.}
Let
$
N^{(1)},\ldots,N^{(M_n)}
$
denote the post-burn-in samples generated by the MH algorithm from the conditional distribution of $N$ given $Y$. Define the exact normalized EM criterion by
$$
Q_n^{\mathrm{EM}}(\theta)
= \frac1n E_{\theta_0}\left[\ell_c(\theta;Y,N)\mid Y
\right],
$$
and its Monte Carlo approximation by
$$
Q_{n,M_n}^{\mathrm{EM}}(\theta)
= \frac{1}{nM_n}\sum_{m=1}^{M_n}\ell_c
\left(\theta;Y,N^{(m)}\right).
$$
The MH chain is assumed to be sufficiently ergodic, and the number of Monte Carlo draws $M_n$ increases with $n$ such that
$$
\sup_{\theta\in\Theta}
\left|
Q_{n,M_n}^{\mathrm{EM}}(\theta)
-
Q_n^{\mathrm{EM}}(\theta)
\right|
=
o_p(n^{-1/2}).
$$
Thus, the Monte Carlo approximation error is asymptotically negligible relative to the statistical $n^{-1/2}$ scale.

\subsection{Analytic Results}

\begin{lemma}[Positive Definiteness]
Under Assumptions (A1) and (A4), the conditional covariance matrix
$\Sigma_n(\theta)$ is positive definite for every
$n$ and every $\theta\in\Theta$. Moreover,
$$
\lambda_{\min}\{\Sigma_n(\theta)\}
\geq
\underline{\sigma}^2\Delta t>0.
$$
\end{lemma}

\begin{proof}
For any nonzero vector $v\in\mathbb{R}^n$,
$$
v^\top\Sigma_n(\theta)v
=\sigma^2\Delta t\|v\|^2
+ \nu^2v^\top\Gamma_n(H)v
+\sigma_J^2
v^\top\operatorname{Diag}(N)v.
$$
The fractional Gaussian noise covariance matrix $\Gamma_n(H)$ is positive semidefinite, and
$\operatorname{Diag}(N)$ is positive semidefinite because
$N_t\geq0$. Hence,
$$
v^\top\Sigma_n(\theta)v
\geq
\sigma^2\Delta t\|v\|^2.
$$
By (A1), $\sigma^2\geq\underline{\sigma}^2>0.$ Therefore, $$ v^\top\Sigma_n(\theta)v
\geq
\underline{\sigma}^2\Delta t\|v\|^2>0.
$$

Thus, $\Sigma_n(\theta)\succ0$
and  $
\lambda_{\min}\{\Sigma_n(\theta)\}
\geq
\underline{\sigma}^2\Delta t.
$
Notice that no uniform upper bound on
$\lambda_{\max}{\Sigma_n(\theta)}$ is required. In particular, the Poisson diagonal component may produce eigenvalues that increase with $n$.
\end{proof}

\begin{lemma}[Stationarity and Ergodicity]
Under Assumptions (A3) and (A4), the observed process
${Y_t}$ is stationary and ergodic.
\end{lemma}

\begin{proof}
Fractional Gaussian noise is stationary for every
$H\in(0,1)$ and, being a nondegenerate stationary Gaussian process, is ergodic; see Beran~\cite{Beran1994}. The sequence
${N_t}$ is iid and therefore stationary and ergodic. The Gaussian jump-size innovations ${Z_t}$ are also iid and ergodic.

By (A4), the fractional Gaussian noise, jump counts, and jump-size innovations are mutually independent. Hence their joint process is stationary and ergodic. Since
$$
Y_t
=
\mu\Delta t
+
X_t
+
\mu_JN_t
+
\sigma_J\sqrt{N_t}Z_t
$$
is a measurable function of this joint process, the observed process ${Y_t}$ is stationary and ergodic.
\end{proof}

\begin{lemma}[Uniform Law of Large Numbers]
Let
\begin{align*}
Q_n(\theta)
&= \frac{1}{n}\ell_n(\theta), \\
\ell_n(\theta)
&= \log p_\theta(Y).
\end{align*}
define
$$
Q(\theta)
=
\lim_{n\to\infty}
E_{\theta_0}[Q_n(\theta)].
$$
Under Assumptions (A1)--(A4),
$$
\sup_{\theta\in\Theta}
|Q_n(\theta)-Q(\theta)|
\xrightarrow{P}0.
$$
\end{lemma}

\begin{proof}
By Lemma 2, the observed process is stationary and ergodic. Therefore, for every fixed $\theta\in\Theta$, the ergodic theorem gives
$$
Q_n(\theta)
\xrightarrow{P}
Q(\theta).
$$
(A2) provides continuity of the likelihood criterion and the required derivatives in $\theta$, while compactness of $\Theta$ follows from (A1). The integrable envelope in (A2), together with the finite moments of the jump component in (A4), provides the domination required to control the criterion uniformly over $\Theta$.

Consequently, the family
$
\{Q_n(\theta):\theta\in\Theta\}
$
is stochastically equicontinuous. The standard uniform convergence argument for extremum estimators therefore yields
$$
\sup_{\theta\in\Theta}
|Q_n(\theta)-Q(\theta)|
\xrightarrow{P}0.
$$
This is the uniform law of large numbers required for consistency; see van der Vaart~\cite{VanDerVaart2000}.
\end{proof}

\begin{lemma}[Identification]
Suppose that two parameter vectors
$$
\theta_1,\theta_2\in\Theta
$$
generate the same marginal distribution for the observed vector $Y$. Then
$$
\theta_1=\theta_2.
$$
\end{lemma}

\begin{proof}
The proof uses the characteristic function of the observed vector and separates the Gaussian long-memory component from the independent jump component.

For a vector $u=(u_1,\ldots,u_n)^\top$, conditional independence of the jump component from the fractional Gaussian component gives
\begin{align*}
\varphi_Y(u)
=
\exp\left(i\mu\Delta t\sum_{t=1}^n u_t\frac12u^\top\left[\sigma^2\Delta t I_n+\nu^2\Gamma_n(H)
\right]u
\right) \\
\times\prod_{t=1}^n
\exp\left\{
\lambda\Delta t
\left[
\exp\left(
i\mu_Ju_t-\frac12\sigma_J^2u_t^2
\right)-1
\right]
\right\}
\end{align*}
Thus, the characteristic exponent has a Gaussian part and a compound-Poisson part.

The compound-Poisson factor is determined by the jump-size distribution $J_t\sim N(\mu_J,\sigma_J^2)$ and the intensity $\lambda$. The uniqueness of the Lévy--Khintchine decomposition implies that the Gaussian component and the jump measure are uniquely determined by the distribution of $Y$. Consequently, $ (\lambda,\mu_J,\sigma_J^2)$ are identified, provided $\lambda>0$ and $\sigma_J^2>0$ as imposed by (A1).
It remains to identify the parameters of the Gaussian long-memory component. After removing the jump contribution, its covariance structure is
$$
C_X(k) = \nu^2\Gamma_H(k), \qquad k\neq0,
$$
where $\Gamma_H(k)$ is the fractional Gaussian noise covariance function. For large $k$,
$$
\Gamma_H(k) \sim H(2H-1)k^{2H-2}.
$$
Hence the rate of decay of the nonzero-lag covariance uniquely determines $H$. Once $H$ is identified, its magnitude determines $\nu^2$. This is consistent with the low-frequency spectral characterization
$$
f_H(\omega)
\sim
C(H)|\omega|^{1-2H},
\qquad
\omega\to0,
$$
used in the long-memory identification theory; see Fox and Taqqu~\cite{FoxTaqqu1986} and Beran~\cite{Beran1994}.

Finally, after $H$, $\nu^2$, $\lambda$, $\mu_J$, and $\sigma_J^2$ have been identified, the marginal variance contribution of the jump component is known. The remaining diagonal Gaussian variance identifies $\sigma^2$ through
$$
\operatorname{Var}(Y_t)
=
\sigma^2\Delta t
+
\nu^2\Gamma_H(0)
+
\lambda(\sigma_J^2+\mu_J^2).
$$
The mean satisfies
$$
E(Y_t)
=
\mu\Delta t+\lambda\Delta t\,\mu_J,
$$
so that, after $\lambda$ and $\mu_J$ have been identified,
$$
\mu
=
\frac{E(Y_t)}{\Delta t}
-
\lambda\mu_J
$$
is uniquely determined.

Therefore all components of
$
\theta=
(\mu,\mu_J,\lambda,\sigma^2,\nu^2,H,\sigma_J^2)^\top
$
are uniquely determined by the distribution of $Y$. Hence
$$
p_{\theta_1}(Y)=p_{\theta_2}(Y)
\quad\Longrightarrow\quad
\theta_1=\theta_2.
$$
\end{proof}

\begin{lemma}[Score Central Limit Theorem]\mbox{}\\
Let
$$
S_n(\theta)
=
\nabla_\theta\ell_n(\theta)
$$
denote the observed-data score. Suppose that the Gaussian long-memory score component satisfies the central limit conditions of Fox and Taqqu~\cite{FoxTaqqu1986} and Dahlhaus~\cite{Dahlhaus1989}, and that the additional contribution generated by the latent compound-Poisson component has finite second moments and satisfies the corresponding joint central limit condition. 
Then
$$
\frac{1}{\sqrt n}S_n(\theta_0) \xrightarrow{d}
N\left(0,I_Y(\theta_0)\right),
$$
where $I_Y(\theta_0)$ is positive definite.
\end{lemma}

\begin{proof}
The score of the observed-data likelihood can be written using Fisher's identity as
$$
S_n(\theta)
=
E_\theta
\left[
\nabla_\theta
\ell_c(\theta;Y,N)
\mid Y
\right].
$$
Consequently, the observed-data score incorporates both the long-memory Gaussian component and the latent jump component.

For the Gaussian component, the spectral density satisfies
$$
f_H(\omega)
\sim
C(H)|\omega|^{1-2H},
\qquad
\omega\to0,
$$
with $H>1/2$. Fox and Taqqu~\cite{FoxTaqqu1986} established consistency and asymptotic normality for likelihood-type estimation in strongly dependent stationary Gaussian sequences under conditions encompassing fractional Gaussian noise. Dahlhaus~\cite{Dahlhaus1989} established asymptotic normality and the limiting Fisher information for maximum-likelihood estimation in long-range-dependent Gaussian processes. These results provide the Gaussian part of the $n^{1/2}$ likelihood asymptotics.

The jump process itself has no long-range dependence. By (A4), its counts and jump-size innovations have finite moments and are independent across time and independent of the fractional Gaussian component. Therefore the jump contribution has finite second-order fluctuations. The additional condition that the latent-jump contribution and the Gaussian score satisfy a joint central limit theorem ensures that their combined contribution is asymptotically Gaussian.

Hence the complete observed-data score satisfies
$$
\frac{1}{\sqrt n}
S_n(\theta_0)
\xrightarrow{d}
N\left(
0,I_Y(\theta_0)
\right).
$$
The important point is that the cited Gaussian results establish the long-memory Gaussian part of the argument; the extension to the MFMJD observed score requires the finite-moment and joint-convergence condition stated in the lemma.
\end{proof}

\begin{lemma}[Observed Information Convergence]\mbox{}\\
Let
$$
J_n(\theta)
=
-\nabla_\theta^2\ell_n(\theta)
$$
denote the observed information matrix. If  
$
\widetilde\theta_n
\xrightarrow{P}
\theta_0,
$ then, under the smoothness and long-memory likelihood conditions above,
$$
\frac1nJ_n(\widetilde\theta_n)
\xrightarrow{P}
I_Y(\theta_0).
$$
\end{lemma}

\begin{proof}
The conditional covariance matrix contains the Toeplitz fractional Gaussian noise component
$\Gamma_n(H).$

For stationary Gaussian processes, the asymptotic behavior of the associated Toeplitz matrices, log-determinants, traces, and quadratic forms provides the basis for likelihood asymptotics. The relevant Toeplitz theory is given by Grenander and Szeg$\H{o}$~\cite{GrenanderSzego1958}, while its application to long-memory likelihoods is developed by Fox and Taqqu~\cite{FoxTaqqu1986} and Dahlhaus~\cite{Dahlhaus1989}.

Under (A1)--(A3), the covariance matrix is positive definite and its dependence on $\theta$ is smooth. Therefore the normalized Gaussian likelihood Hessian converges to its limiting information matrix under the long-memory likelihood regularity conditions.

The jump component contributes through the conditional Gaussian covariance
$
\sigma_J^2\operatorname{Diag}(N)
$ and through the Poisson likelihood. By (A4), the jump variables possess finite moments of every order. Consequently, the normalized jump contribution to the Hessian satisfies the corresponding law of large numbers.

Combining the Gaussian Toeplitz contribution with the finite-moment jump contribution gives
$$
\frac1nJ_n(\widetilde\theta_n)
\xrightarrow{P}
I_Y(\theta_0).
$$
\end{proof}

\subsection{Consistency}

\begin{theorem}[Consistency]\mbox{}\\
Let
$$
\widehat\theta_{n,M_n}
=
\arg\max_{\theta\in\Theta}
Q_{n,M_n}^{\mathrm{EM}}(\theta).
$$
Under Assumptions (A1)--(A5) and Lemmas 1--4,
$$
\widehat\theta_{n,M_n}
\xrightarrow{P}
\theta_0.
$$
\end{theorem}

\begin{proof}
By the triangle inequality,
\begin{align*}
\sup_{\theta\in\Theta}
\left|
Q_{n,M_n}^{\mathrm{EM}}(\theta)-Q(\theta)
\right|
\leq
\sup_{\theta\in\Theta}
\left|
Q_{n,M_n}^{\mathrm{EM}}(\theta)
-
Q_n^{\mathrm{EM}}(\theta)
\right|\\
+
\sup_{\theta\in\Theta}
\left|
Q_n^{\mathrm{EM}}(\theta)-Q(\theta)
\right|.
\end{align*}
The first term converges to zero in probability by (A5), while the second converges uniformly to zero by Lemma 3. Hence,
$$
\sup_{\theta\in\Theta}
\left|
Q_{n,M_n}^{\mathrm{EM}}(\theta)-Q(\theta)
\right|
\xrightarrow{P}0.
$$
By Lemma 4, $Q(\theta)$ has a unique maximizer at $\theta_0$. Therefore, the conditions of the Argmax theorem are satisfied see the standard consistency theorem for extremum estimators, in van der Vaart~\cite{VanDerVaart2000}, 

Hence:
$$
\widehat\theta_{n,M_n}
\xrightarrow{P}
\theta_0.
$$
\end{proof}

\subsection{Asymptotic Normality}

\begin{theorem}[Asymptotic Normality]
Suppose that the conditions of the consistency theorem hold and that the score and observed-information results of Lemmas 5 and 6 hold. Then
$$
\sqrt n
\left(
\widehat\theta_{n,M_n}-\theta_0
\right)
\xrightarrow{d}
N\left(
0,I_Y(\theta_0)^{-1}
\right).
$$
\end{theorem}

\begin{proof}
By the consistency theorem,
$$
\widehat\theta_{n,M_n}
\xrightarrow{P}
\theta_0.
$$
Since $\theta_0$ is an interior point of $\Theta$, the first-order condition holds with probability tending to one:
$$
\nabla_\theta
Q_{n,M_n}^{\mathrm{EM}}
(\widehat\theta_{n,M_n})
=
0.
$$
A Taylor expansion around $\theta_0$ gives
$$
0 =\nabla_\theta
Q_{n,M_n}^{\mathrm{EM}}(\theta_0)
+
\nabla_\theta^2
Q_{n,M_n}^{\mathrm{EM}}(\widetilde\theta_n)
\left(
\widehat\theta_{n,M_n}-\theta_0
\right),
$$
where $\widetilde\theta_n$ lies between
$\widehat\theta_{n,M_n}$ and $\theta_0$.

Therefore,
$$
\sqrt n
\left(
\widehat\theta_{n,M_n}-\theta_0
\right)
=
-
\left[
\nabla_\theta^2
Q_{n,M_n}^{\mathrm{EM}}
(\widetilde\theta_n)
\right]^{-1}
\sqrt n
\nabla_\theta
Q_{n,M_n}^{\mathrm{EM}}(\theta_0).
$$
By (A5), the Monte Carlo approximation error is negligible at the $n^{-1/2}$ scale. Consequently,
$$
\sqrt n
\nabla_\theta
Q_{n,M_n}^{\mathrm{EM}}(\theta_0)
=
\frac1{\sqrt n}
\nabla_\theta\ell_n(\theta_0)
+
o_p(1).
$$
By Lemma 5,
$$
\sqrt n
\nabla_\theta
Q_{n,M_n}^{\mathrm{EM}}(\theta_0)
\xrightarrow{d}
N\left(
0,I_Y(\theta_0)
\right).
$$
Furthermore, consistency implies
$$
\widetilde\theta_n
\xrightarrow{P}
\theta_0.
$$
Lemma 6 therefore gives
$$
-\nabla_\theta^2 Q_{n,M_n}^{\mathrm{EM}}
(\widetilde\theta_n)\xrightarrow{P}I_Y(\theta_0).
$$
Since $I_Y(\theta_0)$ is nonsingular,
$$
\left[-\nabla_\theta^2Q_{n,M_n}^{\mathrm{EM}}
(\widetilde\theta_n)\right]^{-1}\xrightarrow{P}
I_Y(\theta_0)^{-1}.
$$
Applying Slutsky's theorem yields
$$
\sqrt n\left(\widehat\theta_{n,M_n}-\theta_0 \right)\xrightarrow{d}
N\left(0,I_Y(\theta_0)^{-1}\right).
$$
The extremum-estimation and Taylor-expansion arguments are standard; see van der Vaart~\cite{VanDerVaart2000} and Ibragimov and Has'minskii~\cite{IbragimovHasminskii2013}.
\end{proof}


\section{Data}

The empirical analysis of this model employs daily closing prices of the Helsinki stock index(OMXH25) from January 2021 to June 2025. Let $S_t$ denote the closing price on trading day $t$. The continuously compounded daily return is defined by
\begin{equation}
Y_t=\log(S_t)-\log(S_{t-1})
=\log\left(\frac{S_t}{S_{t-1}}\right).
\end{equation}
After the data cleaning procedure, the resulting sample contains
$n=1171$ daily observations.

\begin{table}[h!]
\centering
\caption{Descriptive statistics of daily logarithmic returns}
\label{tab:descriptive}
\resizebox{\textwidth}{!}{%
\begin{tabular}{lrrrrrrr}
\hline
Variable & $N$ & Mean & Std. Dev. & Minimum & Maximum &
Skewness & Kurtosis \\
\hline
Daily log returns
& 1171
& $-2.5971\times10^{-5}$
& 0.010016
& -0.048204
& 0.040978
& -0.429975
& 2.605813 \\
\hline
\end{tabular}%
}
\end{table}

Table~\ref{tab:descriptive} presents the descriptive statistics of the return series. The mean daily return is
$-2.5971\times10^{-5}$, while the standard deviation is $0.010016$.
The minimum and maximum returns are $-0.048204$ and $0.040978$,
respectively. The return distribution exhibits moderate negative
skewness, with a skewness coefficient of $-0.429975$, indicating
greater asymmetry in the negative tail. The reported kurtosis is
$2.605813$, which is below the Gaussian benchmark of 3. These
characteristics, together with the substantial daily price movements
illustrated in Figures~\ref{fig:helsinki_price} and
\ref{fig:helsinki_returns}, provide empirical motivation for a model
that permits both persistent dependence and discontinuous movements.
The fractional component accommodates dependence through
$H>1/2$, while the compound Poisson component explicitly represents
discrete market shocks.
\begin{figure}[H]
\centering
\includegraphics[width=0.85\textwidth]{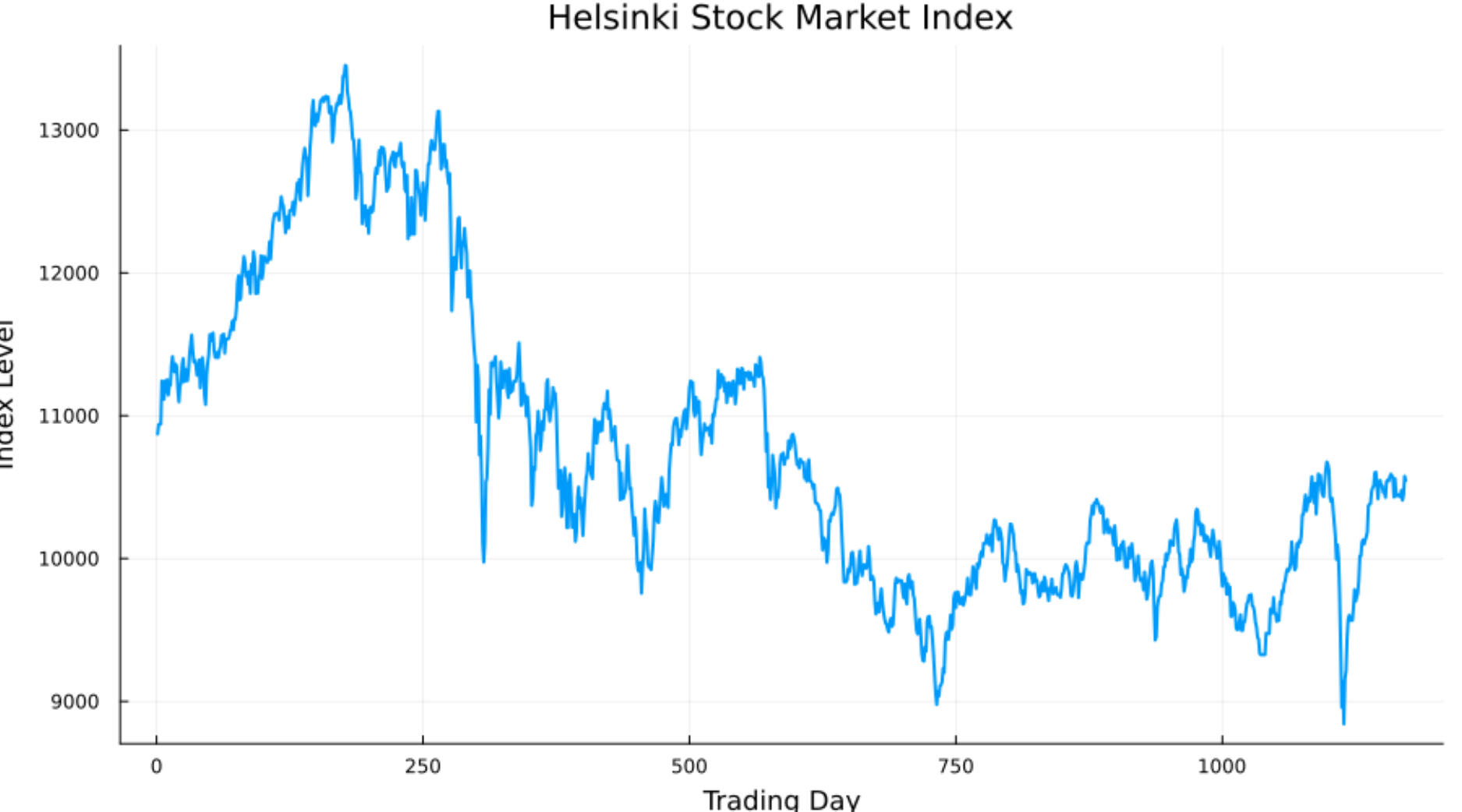}
\caption{Daily closing prices of the Helsinki stock market index,
January 2021--June 2025.}
\label{fig:helsinki_price}
\end{figure}

\begin{figure}[H]
\centering
\includegraphics[width=0.85\textwidth]{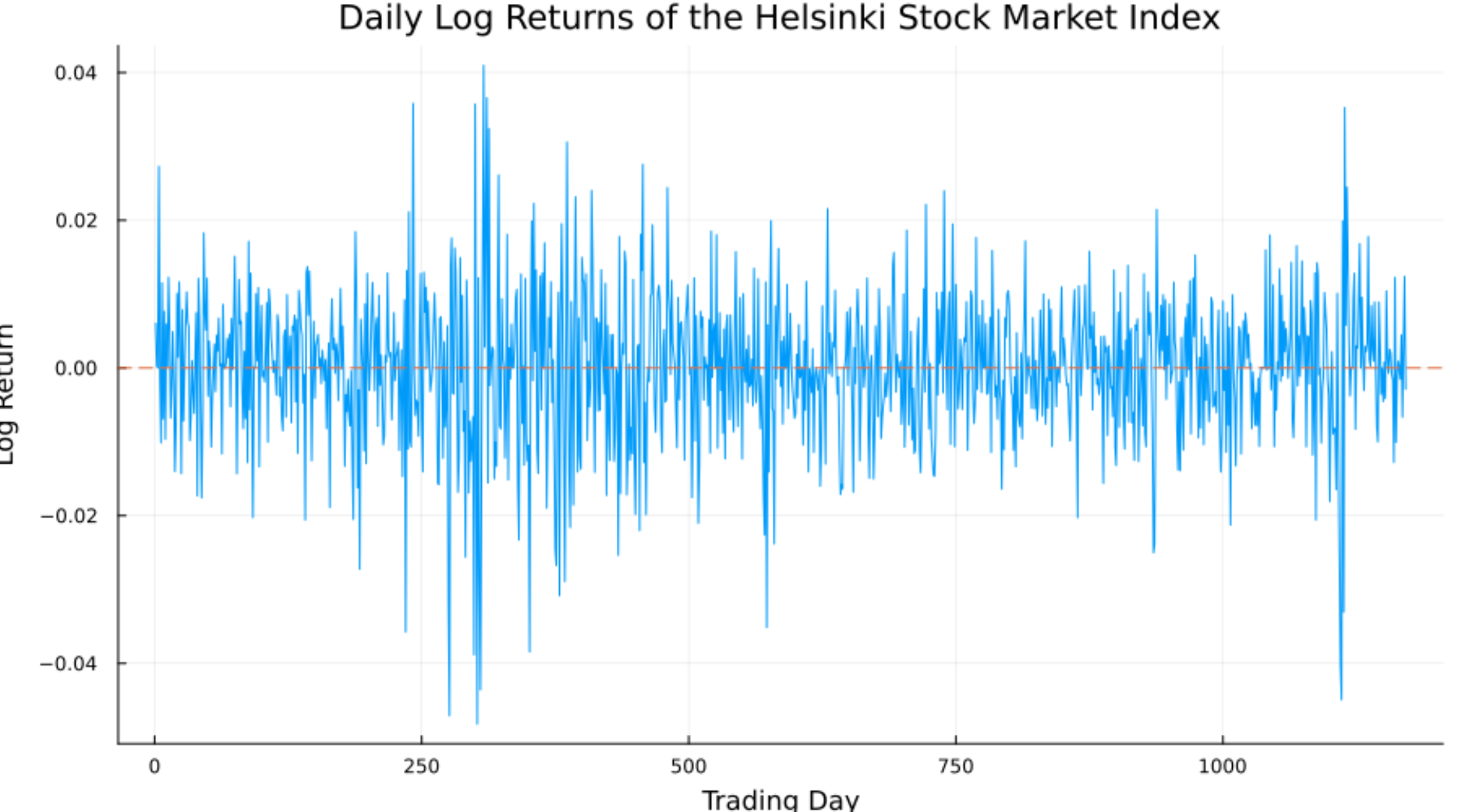}
\caption{Daily logarithmic returns of the Helsinki stock market index,
January 2021--June 2025.}
\label{fig:helsinki_returns}
\end{figure}
\section{Results and Simulation Study}

\subsection{Jump Detection and Parameter Estimation}

We implement MFMJD model on \(1171\) daily return observations of the Helsinki Stock Index spanning from January 2021 to June 2025, using annualized time steps (\(\Delta t = \frac{1}{252}\)). both Jump detection and Parameter estimation was done using EM algorithm MCMC with the help of Julia. A location is considered to have a jump if has 95 \% probability of Jumps according to the Julia code.  The estimated parameters are as follows:
\[
\mu = 2.3406 \times 10^{-2} , \quad 
\sigma^2_J = 3.6052  \times 10^{-2}, \quad 
\lambda = 0.8085,
\]
\[
\sigma^2 = 1.5729 \times 10^{-2}, \quad 
\nu^2 = 0.001398, \quad 
H = 0.76073, \mu_J = 1.9576 \times 10^{-2} 
\]
The estimated parameters, provide compelling evidence of long range dependency within the Nordic region during this macroeconomic cycle.
The estimated Hurst exponent (\(H = 0.7607\)) sits heavily above the classical random walk benchmark (\(H = 0.5\)), indicating strong long-range persistence and memory effects. This behavior matches the prolonged, trend-heavy cycles of the index as it adjusted to post-pandemic macro realignments and European trade disruptions. This result is consistent with the fractional dynamics introduced by Mandelbrot and Van Ness (1968). Crucially, our framework detected \(132\) distinct jump locations, resulting in an annualized jump intensity \(\lambda = 0.8085\). The conditional variance of these jumps (\(\sigma_j^2 = 0.0361\)) vastly exceeds the baseline continuous variance (\(\sigma^2 = 0.0157\)). This structural divergence proves that market risk over this window was heavily asymmetric, driven primarily by discrete, block-information shocks rather than continuous diffusion. The Figure~\ref{fig:jump_detection} shows the price path with detected jump location
\begin{figure}[H]
\centering
\includegraphics[width=0.9\linewidth]{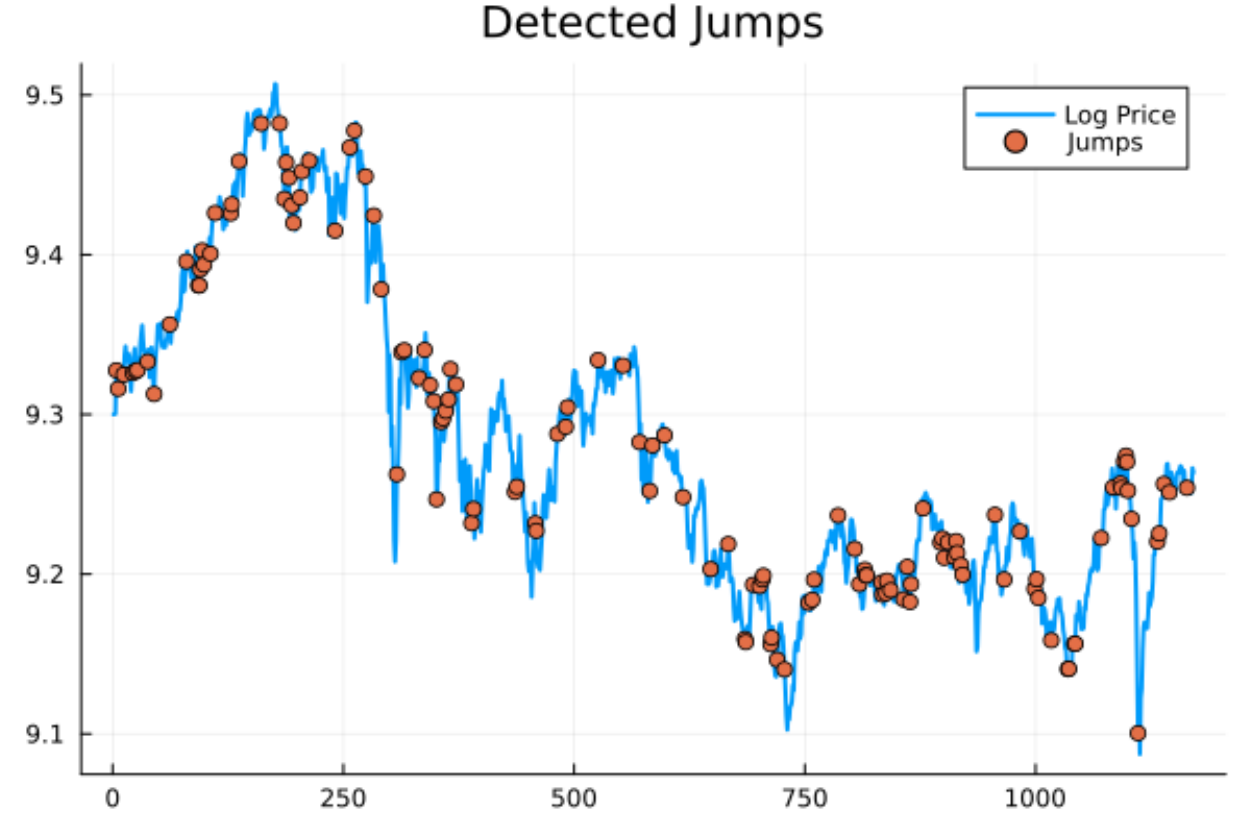}
\caption{Detected Jumps: Price path with detected jump location.}
\label{fig:jump_detection}
\end{figure}

\subsection{Model Fit}

To evaluate the generative power of the estimated Mixed Fractional Merton Jump Diffusion (MFMJD) parameters, we simulate 30 independent asset trajectories across a 1171-day horizon using the annualized empirical parameters. Figure~\ref{fig:model_fit} illustrates that the historical trajectory of the Helsinki Stock Index (solid black line) is centrally embedded within the simulated path envelope.This shows that the simulated trajectories clearly replicate the stylized facts of the empirical data. This confirms the effectiveness of combining long-memory processes \cite{Mandelbrot1968} with jump-diffusion dynamics \cite{Merton1976}. 
\begin{figure}[H]
\centering
\includegraphics[width=0.9\linewidth]{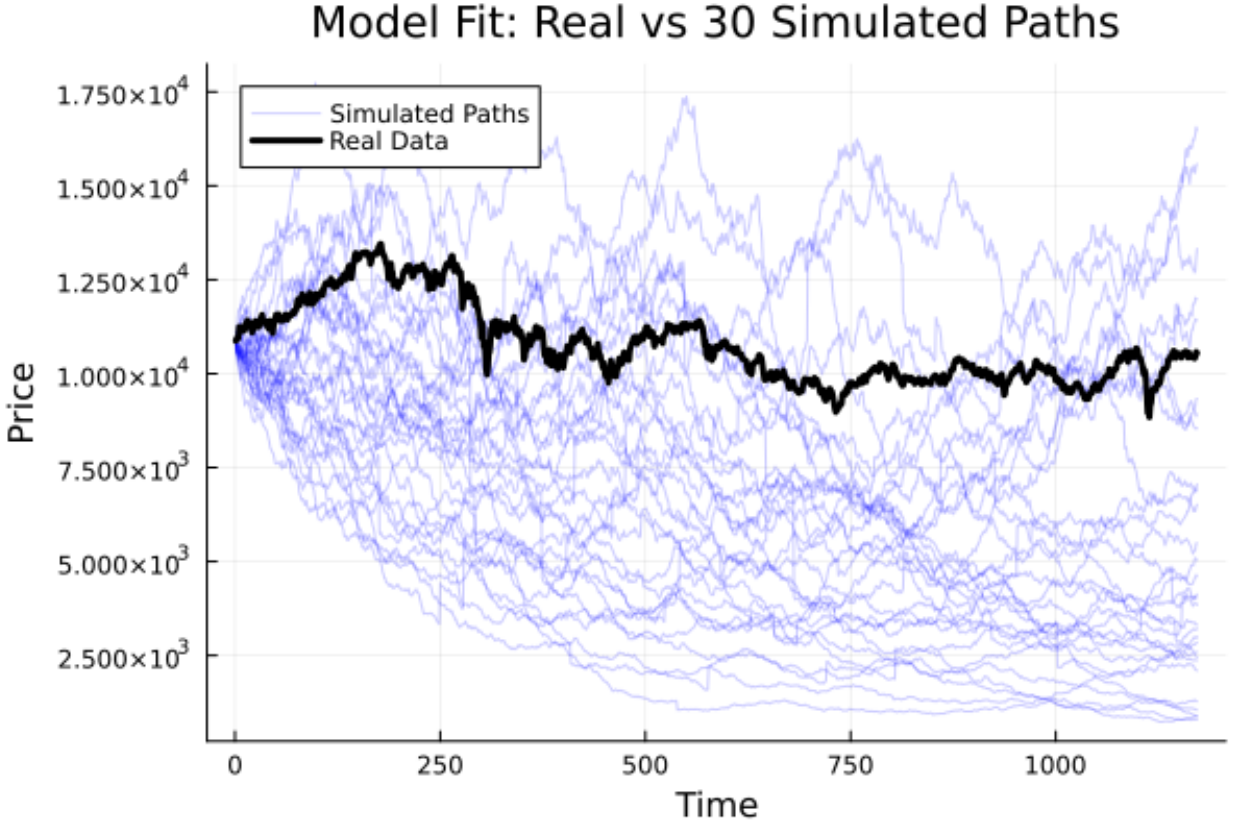}
\caption{Model fit: real path versus 30 simulated paths.}
\label{fig:model_fit}
\end{figure}

\subsection{Residual Analysis}

\paragraph{Distributional Properties.}
To assess the empirical robustness of our EM-MCMC parameter estimation, we analyze the resulting model residuals. The empirical distribution of the residuals demonstrates highly favorable diagnostic properties. The distribution is tightly centered around a mean of zero, confirming the absence of residual deterministic drift or systematic estimation bias.
Figure~\ref{fig:residual_hist} presents the residual histogram. 

\begin{figure}[H]
\centering
\includegraphics[width=0.8\linewidth]{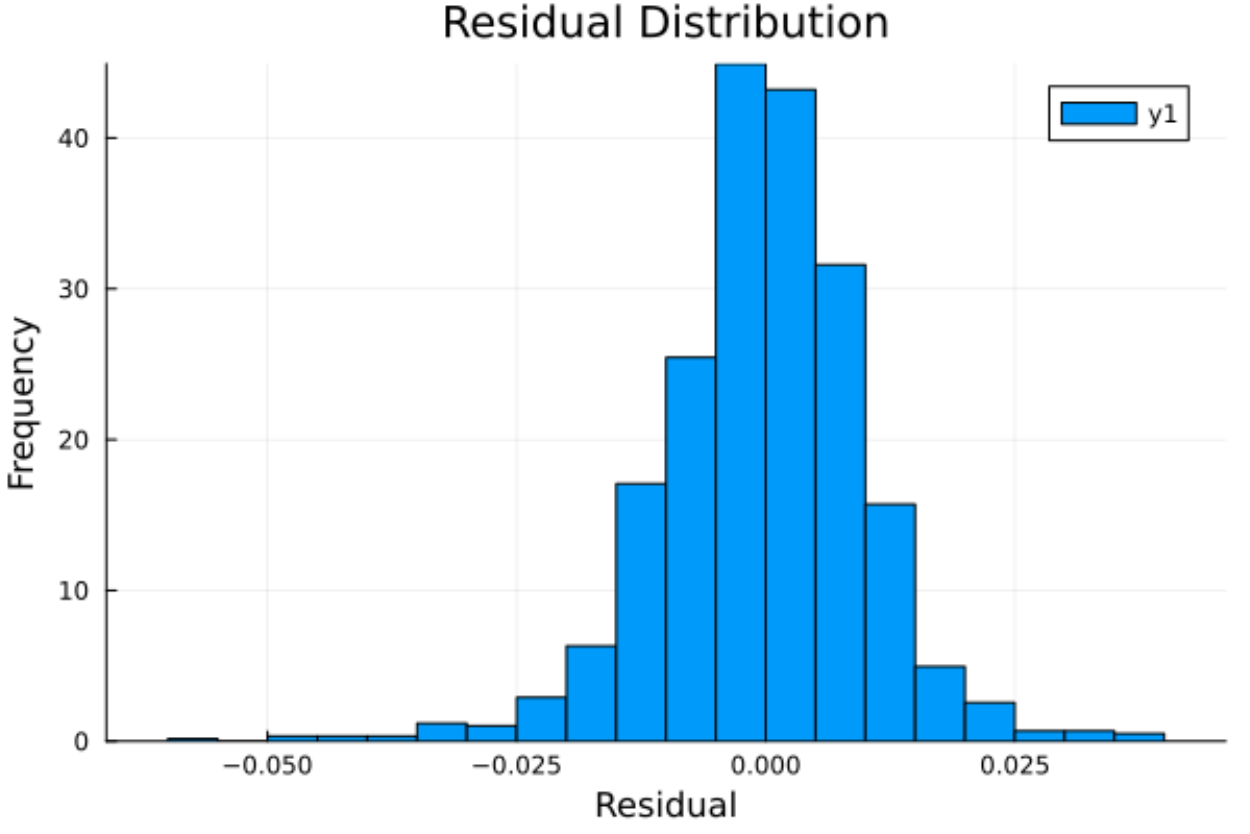}
\caption{Residual distribution.}
\label{fig:residual_hist}
\end{figure}

The distribution exhibits leptokurtic behavior and slight negative skewness, consistent with empirical findings in financial returns \cite{mandelbrot1963}. The presence of fat tails validates the inclusion of jump components, which capture extreme events beyond Gaussian assumptions.

\paragraph{Residual over time.}
To evaluate the temporal stability and independence of the model parameters,Figure~\ref{fig:residual_time} tracks the estimated residuals over the 1171-day trading horizon (January 2021 to June 2025).

\begin{figure}[H]
\centering
\includegraphics[width=0.9\linewidth]{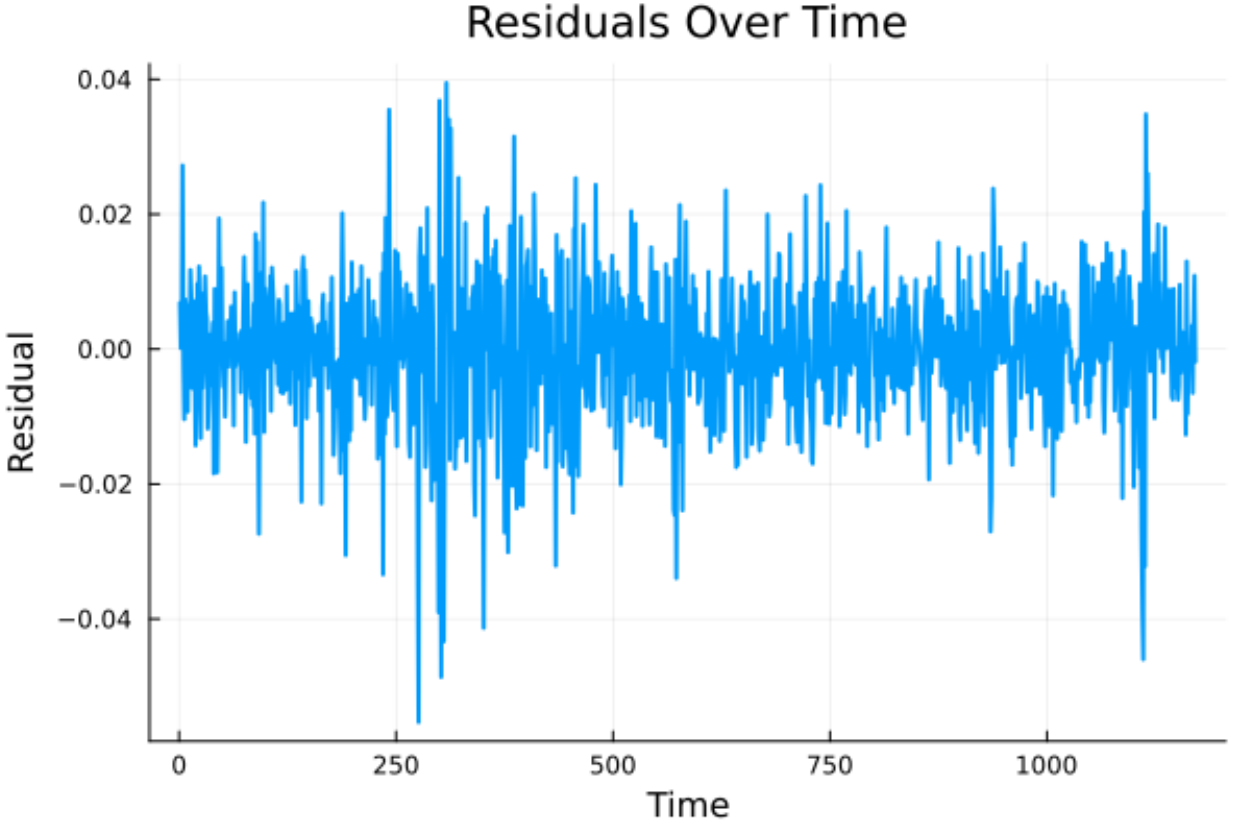}
\caption{Residuals over time.}
\label{fig:residual_time}
\end{figure}

The residuals exhibit structural stationarity, fluctuating uniformly around a zero mean baseline. This absence of persistent local trends confirms that the fractional drift architecture, parameterized by \(H = 0.7607\), successfully exhausted the long-range serial dependence inherent in the empirical index returns.

\subsection{Goodness-of-Fit}

The goodness of fit of the proposed model is evaluated using the
mean squared error (MSE) between the observed and simulated index
prices. Since multiple simulated paths are generated, the reported
measure is the average MSE across the $M$ simulated paths, defined as
\begin{equation}
\operatorname{Average\ MSE}
=
\frac{1}{M}
\sum_{j=1}^{M}
\left[
\frac{1}{n}
\sum_{t=1}^{n}
\left(
S_t^{\mathrm{real}}-S_{t,j}^{\mathrm{sim}}
\right)^2
\right],
\end{equation}
where $S_t^{\mathrm{real}}$ denotes the observed index price and
$S_{t,j}^{\mathrm{sim}}$ denotes the simulated price at time $t$ for
the $j$th simulated path.

Using the simulated ensemble, the resulting average MSE is
\[
\operatorname{Average\ MSE}
=6.42498\times10^{4}.
\]
Because the MSE is expressed in squared price units, its square root
provides a more interpretable measure of the typical pricing
deviation. The corresponding RMSE is approximately
\[
\operatorname{RMSE}
=\sqrt{6.42498\times10^{4}}
\approx253.48.
\]
The observed OMX25 prices range from approximately
9,135 to 13,800 over the sample period. Hence, an RMSE of about 253.48 index points represents a relatively small deviation compared
with the overall level of the index, corresponding to approximately
2.2\% relative to a representative price level near the centre of the observed range. This indicates that the simulated paths generated
from the estimated MFMJD parameters provide a reasonably close representation of the observed price dynamics.
The obtained MSE indicate that the fitted values produced by the proposed EM--MH estimation procedure closely match the observed log-return series. Combined with the model fit and residual diagnostics, this result suggests that the proposed MFMJD model provides an adequate representation of both continuous market dynamics and discontinuous jump behavior.

\section{Summary and suggestion for further studies }

The MFMJD model demonstrates strong empirical performance by integrating fractional dynamics with jump diffusion. The results confirm that the model captures long-range dependence, heavy tails, and discrete shocks, providing a robust framework for financial modeling consistent with established literature  \cite{Merton1976, mandelbrot1963, Mandelbrot1968}
Specifically, the MFMJD model provides a robust representation of the Helsinki Stock Index by integrating fractional dynamics with jump diffusion. The estimated parameters, simulation results,and residual analysis collectively confirm that the model captures both persistent behavior and extreme market movements, leaving residuals that behave as pure stochastic noise. The model can be applied to any stock market data which exhibits the same behaviour. 
\paragraph{} However,  since the empirical analysis in this study is restricted to one market data set "the Helsinki Stock Index". The researchers suggest a further research with Validation across multiple financial markets and asset classes to further strengthen the general applicability and robustness of the framework.
Also, since this research focuses mainly on empirical validation and establishment of facts, we suggest that further research should be conducted on forecasting and prediction of the stock prices.

\bibliographystyle{siam}
\bibliography{Biblog1}
\end{document}